\documentclass{eptcs}
\providecommand{\event}{QPL 2011}
\usepackage{amssymb}
\usepackage{amsmath}
\usepackage{stmaryrd}
\usepackage{prooftree}
\usepackage{paralist}

\newtheorem{proposition}{Proposition}
\newtheorem{theorem}{Theorem}
\newtheorem{definition}{Definition}
\newenvironment{proof}{\textit{Proof:}}{}
\newcommand{\qed}{\hfill$\Box$}
\newcommand{\mktype}[1]{\mathsf{#1}}

\newcommand{\Qbit}{\mktype{Qbit}}

\newcommand{\ctxt}[2]{{#1}[#2]}

\newcommand{\Prob}[1]{\boxplus_{#1}}
\newcommand{\ms}[1]{|{#1}|^{2}}
\newcommand{\cnfig}[3]{({#1};{#2};{#3})}

\newcommand{\ptrns}[3]{{#1}\stackrel{#2}{\longrightarrow}{#3}}

\newcommand{\typed}[2]{{#1}\mathrel{\!:\!}{#2}}

\newcommand{\ket}[1]{|#1\rangle}
\newcommand{\nil}{\mathbf{0}}
\renewcommand{\parallel}{\mathbin{\mid}}
\newcommand{\inp}[2]{{#1}?{[#2]}}
\newcommand{\outp}[2]{{#1}!{[#2]}}

\renewcommand{\vec}[1]{\widetilde{#1}}
\newcommand{\new}{\mathsf{new}\ }
\newcommand{\qbit}{\mathsf{qbit}\ }
\newcommand{\bit}{\mathsf{bit}}

\newcommand{\chant}[1]{\widehat{~}[#1]}

\newcommand{\qgate}[1]{\mathsf{#1}}
\newcommand{\pname}[1]{\mathit{#1}}
\newcommand{\action}[1]{\{#1\}}
\newcommand{\trans}{\mathbin{*\!\!=}}
\newcommand{\sep}{\,.\,}

\newcommand{\measure}{\mathsf{measure}\ }

\newcommand{\matr}[4]{\begin{pmatrix}{#1}&{#2}\\{#3}&{#4}\end{pmatrix}}

\newcommand{\qstore}[2]{[#1 \mapsto #2]}

\newcommand{\transition}[1]{\stackrel{#1}{\longrightarrow}}

\newcommand{\ptrans}[1]{\stackrel{#1}\rightsquigarrow}
\newcommand{\weaktrans}[1]{\stackrel{#1}{\Longrightarrow}}
\newcommand{\opttrans}[1]{\stackrel{#1}{\longrightarrow}^{+}}

\newcommand{\cfgprob}[1]{#1 \bullet}
\newcommand{\cfgprobsum}{\boxplus}
\newcommand{\cfgplus}[1]{\cfgprobsum~\cfgprob{#1}}
\newcommand{\cfgdistsum}{\oplus}

\newcommand{\Dist}[2]{\cfgdistsum_{#1}~#2~}

\newtheorem{example}{Example}

\newcommand{\subst}[2]{\{{#1}/{#2}\}}

\newcommand{\bra}[1]{\langle#1|}

\newcommand{\ketbra}[2]{\ket{#1}\bra{#2}}

\newcommand{\parcomp}{~\Vert~}
\newcommand{\rhoe}{\rho_E}

\newcommand{\gX}{\qgate{X}}

\newcommand{\gH}{\qgate{H}}

\newcommand{\ltrm}[3]{\lambda{#1}\bullet{#2}; {#3}}
\newcommand{\ltrmshort}[2]{\lambda{#1}\bullet{#2}}
\newcommand{\states}{\ensuremath{\mathcal{S}}}
\newcommand{\nstates}{\ensuremath{\mathcal{S}_n}}
\newcommand{\pstates}{\ensuremath{\mathcal{S}_p}}

\newcommand{\pbsim}{\leftrightarroweq}
\newcommand{\fpbsim}{\leftrightarroweq^c}

\renewcommand{\vec}[1]{\widetilde{#1}}
\newcommand{\trace}{\mathrm{tr}}

\newcommand{\pidentity}{\pname{Identity}}

\newcommand{\pteleport}{\mathit{Teleport}}

\newcommand{\palice}{\pname{Alice}}
\newcommand{\pbob}{\pname{Bob}}
\newcommand{\pbobrec}{\pname{BobRec}}
\newcommand{\pbobcorr}{\pname{BobCorr}}
\newcommand{\pnoise}{\pname{Noise}}
\newcommand{\pnoisernd}{\pname{NoiseRnd}}
\newcommand{\pnoiseerr}{\pname{NoiseErr}}

\newcommand{\pQECC}{\pname{QECC}}

\newcommand{\prnd}{\pname{Rnd}}
\newcommand{\pflip}{\pname{Flip}}
\newcommand{\pbitflip}{\pname{BitFlip}}

\makeatletter
\def\@fnsymbol#1{\ifcase#1\or *\or **\or \ddagger\or \mathchar "278\or \mathchar "27B\or \|\or **\or \dagger\dagger \or \ddagger\ddagger \else\@ctrerr\fi\relax}
\makeatother

\title{Analysis of a Quantum Error Correcting Code using Quantum
  Process Calculus}
\author{Timothy A. S. Davidson
\institute{Department of Computer Science\\
University of Warwick, UK}
\and
Simon J. Gay
\institute{School of Computing Science\\
University of Glasgow, UK}
\and 
Rajagopal Nagarajan\thanks{Partially supported by ``Process Algebra
Approach to Distributed Quantum Computation and Secure Quantum
Communication'', Australian Research Council Discovery Project
DP110103473.}
\institute{Department of Computer Science\\
University of Warwick, UK}
\and
Ittoop Vergheese Puthoor\thanks{Supported by a Lord Kelvin / Adam
  Smith Scholarship from the University of Glasgow.} 
\institute{School of Computing Science and\\
School of Physics and Astronomy\\
University of Glasgow, UK}
} 
\def\titlerunning{Analysis of a Quantum Error Correcting Code using Quantum
  Process Calculus}
\def\authorrunning{T. A. S. Davidson, S. J. Gay, R. Nagarajan and
  I. V. Puthoor}

\begin{document}
\maketitle

\begin{abstract}
We describe the use of quantum process calculus to describe and
analyze quantum communication protocols, following the successful
field of \emph{formal methods} from classical computer science. The
key idea is to define two systems, one modelling a protocol and one
expressing a specification, and prove that they are
\emph{behaviourally equivalent}. We summarize the necessary theory in
the process calculus CQP, including the crucial result that
equivalence is a \emph{congruence}, meaning that it is preserved by
embedding in any context. We illustrate the approach by
analyzing two versions of a quantum error correction system.
\end{abstract}

\section{Introduction}
\label{sec-intro}
Quantum process calculus is a generic term for a class of formal
languages with which to describe and analyze the behaviour of systems
that combine quantum and classical computation and communication.
Quantum process calculi have been developed as part of a programme to
transfer ideas from the field of \emph{formal methods}, well
established within classical computer science, to quantum systems. The
field of formal methods provides theories, methodologies and tools for
verifying the correctness of computing systems, usually systems that
involve concurrent, communicating components. The motivation for
developing quantum formal methods is partly to provide a conceptual
understanding of concurrent, communicating quantum systems, and partly
to support the future development of tools for verifying the
correctness of practical quantum technologies such as cryptosystems.

Our own approach is based on a particular quantum process calculus
called Communicating Quantum Processes (CQP), developed by Gay and
Nagarajan \cite{Gay2005}. Recent work on CQP has addressed the
question of defining \emph{behavioural equivalence} between processes,
which formalizes the idea of observational indistinguishability. The
aim is to support the following methodology for proving correctness of
a system. First, define $\pname{System}$, a process that models the
system of interest. Second, define $\pname{Specification}$, a simpler
process that directly expresses the desired behaviour of
$\pname{System}$. Third, prove that these two processes are
equivalent, meaning indistinguishable by any observer: $\pname{System}
\cong \pname{Specification}$. This approach works best when the notion
of equivalence is a \emph{congruence}, meaning that it is preserved by
inclusion in any environment. While there have been several attempts
to define a congruence for a quantum process calculus, the problem has
only recently been solved: for CQP, reported in Davidson's PhD thesis
\cite{DavidsonThesis}, and independently for qCCS by Feng \emph{et
  al.} \cite{Feng2011}.

The present paper begins in Section~\ref{sec:CQP} by reviewing the
language of CQP and illustrating it with a model of a quantum error
correcting code. Section~\ref{sec:equivalence} then summarizes the
theory of behavioural equivalence for CQP, which has not previously
been published other than in Davidson's thesis, and applies it to the
error correcting code. Section~\ref{sec:example} analyzes the system
in the presence of errors that cannot be corrected. Finally,
Section~\ref{sec:conclusion} concludes with an indication of
directions for future work. The contributions of the paper are the
first publication of the definition of a congruence for CQP, and the
application of CQP congruence to examples beyond the teleportation and
superdense coding systems that have been considered previously.

\paragraph*{Related Work}
Lalire \cite{Lalire2006} defined a probabilistic branching
bisimilarity for the process calculus QPAlg, based on the branching
bisimilarity of van Glabbeek and Weijland \cite{Glabbeek1996}, but it
was not preserved by parallel composition.
Feng et al.\ \cite{Feng2006} developed qCCS and defined strong and
weak probabilistic bisimilarity. Their equivalences are preserved by
parallel composition with processes that do not change the quantum
context. A later version of qCCS \cite{Ying2009} excluded classical
information and introduced the notion of approximate bisimilarity as a
way of quantifying differences in purely quantum process behaviour.
Their strong reduction-bisimilarity is a congruence is not sufficient
for the analysis of most interesting quantum protocols, as the
language does not include a full treatment of measurement. In recent
work, Feng et al.\ \cite{Feng2011} define a new version of qCCS and
prove that weak bisimilarity is a congruence. They apply their result
to quantum teleportation and superdense coding. The details of their
equivalence relation are different from our full probabilistic
branching bisimilarity, and a thorough comparison awaits further work.

The work presented in this paper contrasts with previous work on
model-checking quantum systems. The QMC (Quantum Model-Checker) system
\cite{Gay2008,Gay2010a} is able to verify that a quantum protocol
satisfies a specification expressed in a quantum logic, by
exhaustively simulating every branch of its behaviour. The use of
logical formulae is known as \emph{property-oriented} specification,
in distinction to the \emph{process-oriented} specifications
considered in the present paper. Because of the need for efficient
simulation, QMC uses the stabilizer formalism \cite{aaronson} and is
limited to Clifford group operations. Nevertheless, this is sufficient
for the analysis of a simple error correcting code, and such an
analysis appears in \cite{Gay2010a}. There are two main advantages of
the process calculus approach. First, because we are using
pen-and-paper reasoning rather than computational simulation, there is
no restriction to stabilizer states. Second, the fact that equivalence
is a congruence means that we can use equational reasoning to deduce
further equivalences, whereas in the model-checking approach we only
obtain the particular fact that is checked. The disadvantage of the
process calculus approach is that, unlike the situation for classical
process calculus, equivalence-checking has not yet been automated.

\section{Communicating Quantum Processes (CQP)}
\label{sec:CQP}
\label{sec-CQP}
CQP \cite{Gay2005} is a process calculus for formally defining the
structure and behaviour of systems that combine quantum and classical
communication and computation. It is based on pi-calculus
\cite{Milner1999,Milner1992}, with the addition of primitive
operations for quantum information processing. The general picture is
that a system consists of a number of independent components, or
\emph{processes}, which can communicate by sending data along
\emph{channels}. In particular, qubits can be transmitted on channels.
One of the distinctive features of CQP is its type system, which
ensures that operations can only be applied to data of the appropriate
type. The type system is also used to enforce the view of qubits as
physical resources, each of which has a unique owning process at any
given time. If a qubit is send from $A$ to $B$, then ownership is
transferred and $A$ can no longer access it. Although typing is
important, we will not discuss it in detail in the present paper;
however, our CQP definitions will include type information because it
usually forms useful documentation. Also, in the present paper, we
will not give a full formal definition of the CQP language. Instead,
in the next section, we will explain it informally in relation to our
first model of a quantum error correction system.

\subsection{Error Correction: A First Model}
Our model of a quantum error correction system consists of three
processes: $\palice$, $\pbob$ and $\pnoise$. $\palice$ wants to send a
qubit to $\pbob$ over a noisy channel, represented by $\pnoise$. She
uses a simple error correcting code based on threefold repetition
\cite [Chapter 10]{Nielsen2000}. This code is able to correct a single
bit-flip error in each block of three transmitted qubits, so for the
purpose of this example, in each block of three qubits, $\pnoise$
either applies $\gX$ to one of them or does nothing. $\pbob$ uses the
appropriate decoding procedure to recover $\palice$'s original qubit.
The CQP definition of $\palice$ is as follows.
\[
\begin{array}{rcl}
  \multicolumn{3}{l}{\pname{Alice}(\typed{a}{\chant{\Qbit}},\typed{b}{\chant{\Qbit,\Qbit,\Qbit}})
  =}\\
& & (\qbit y,z)
  \inp{a}{\typed{x}{\Qbit}}\sep\action{x,z\trans\qgate{CNot}}\sep\action{x,y\trans\qgate{CNot}}\sep\outp{b}{x,y,z}\sep\nil
\end{array}
\]
$\palice$ is parameterized by two channels, $a$ and $b$. In order to
give $\palice$ a general definition independent of the qubit to be
sent to $\pbob$, she will receive the qubit on channel $a$. The type
of $a$ is $\chant{\Qbit}$, which is the type of a channel on which
each message is a qubit. Channel $b$ is where $\palice$ sends the
encoded qubits. Each message on $b$ consists of three qubits, as
indicated by the type $\chant{\Qbit,\Qbit,\Qbit}$.

The right hand side of the definition specifies $\palice$'s
behaviour. The first term, $(\qbit y,z)$, allocates two fresh qubits,
each in state $\ket{0}$, and gives them the local names $y$ and
$z$. Then follows a sequence of terms separated by dots. This
indicates temporal sequencing, from left to
right. $\inp{a}{\typed{x}{\Qbit}}$ specifies that a qubit is received
from channel $a$ and given the local name $x$. The term
$\action{x,z\trans\qgate{CNot}}$ specifies that the $\qgate{CNot}$
operation is applied to qubits $x$ and $z$; the next term is
similar. These operations implement the threefold repetition code: if
the intial state of $x$ is $\ket{0}$ (respectively, $\ket{1}$) then
the state of $x,y,z$ becomes $\ket{000}$ (respectively,
$\ket{111}$). In general, of course, the initial state of $x$ may be a
superposition, and then so will be the final state of
$x,y,z$. Finally, the term $\outp{b}{x,y,z}$ means that the qubits
$x,y,z$ are sent as a message on channel $b$. The term $\nil$ simply
indicates termination.



We model a noisy quantum channel by the process $\pnoise$, which
receives three qubits from channel $b$ (connected to $\palice$) and
sends three (possibly corrupted) qubits on channel $c$ (connected to
$\pbob$). $\pnoise$ has four possible actions: do nothing, or apply
$\gX$ to one of the three qubits. These actions are chosen with equal
probability. We produce probabilistic behaviour by introducing fresh
qubits in state $\ket{0}$, applying $\gH$ to put them into state
$\frac{1}{\sqrt{2}}(\ket{0}+\ket{1})$, and then measuring in the
standard basis. The definition of $\pnoise$ is split into two
sub-processes, of which the first, $\pnoisernd$, produces two
random classical bits and sends them to the second, $\pnoiseerr$, on
channel $p$. This programming style, using internal messages instead
of assignment to variables, is typical of pi-calculus. 
\[
\begin{array}{rcl}
\pnoisernd(\typed{p}{\chant{\bit,\bit}})
& = & (\qbit
u,v)\action{u\trans\gH}\sep\action{v\trans\gH}\sep\outp{p}{\measure
  u,\measure v}\sep\nil
\end{array}
\]

The process $\pnoiseerr$ receives three qubits from channel $b$, and
two classical bits from channel $p$. It interprets the classical bits,
locally named $j$ and $k$, as instructions for corrupting the
qubits. This uses appropriate Boolean combinations of $j$
and $k$ to construct conditional quantum operations such as
$\gX^{j\overline{k}}$. 
\[
\begin{array}{rcl}
\multicolumn{3}{l}{\pnoiseerr(\typed{b}{\chant{\Qbit,\Qbit,\Qbit}},\typed{p}{\chant{\bit,\bit}},\typed{c}{\chant{\Qbit,\Qbit,\Qbit}})
  =}\\
& & \inp{b}{\typed{x}{\Qbit},\typed{y}{\Qbit},\typed{z}{\Qbit}}\sep\inp{p}{\typed{j}{\bit},\typed{k}{\bit}}\sep\action{x\trans\gX^{jk}}\sep\action{y\trans\gX^{j\overline{k}}}\sep\action{z\trans\gX^{\overline{j}k}}\sep\outp{c}{x,y,z}\sep\nil 
\end{array}
\]

The complete $\pnoise$ process consists of $\pnoisernd$ and
$\pnoiseerr$ in parallel, indicated by the vertical bar. Channel $p$
is designated as a private local channel; this is specified by 
$(\new p)$. This construct comes from pi-calculus, where it can be used to
dynamically create fresh channels, but here we are using it in the
style of older process calculi such as CCS, to indicate a channel with
restricted scope. Putting $\pnoisernd$ and
$\pnoiseerr$ in parallel means that the output on $p$ in $\pnoisernd$
synchronizes with the input on $p$ in $\pnoiseerr$, so that data is
transferred. 
\[
\begin{array}{rcl}
\pnoise(\typed{b}{\chant{\Qbit,\Qbit,\Qbit}},\typed{c}{\chant{\Qbit,\Qbit,\Qbit}})
& = & (\new p)(\pnoisernd(p) \parallel \pnoiseerr(b,p,c))
\end{array}
\]


$\pbob$ consists of $\pbobrec$ and $\pbobcorr$, where $\pbobrec$
receives the qubits and measures the error syndrome, and $\pbobcorr$
applies the appropriate correction. An internal channel $p$ is used to
transmit the result of the measurement, as well as the original
qubits, again in pi-calculus style. After correcting the error in the
group of three qubits, $\pbobcorr$ reconstructs a quantum state in
which qubit $x$ has the original state received by $\palice$ and is
separable from the auxiliary qubits. Finally, $\pbobcorr$ outputs $x$
on channel $d$.
\[
\begin{array}{rcl}
\multicolumn{3}{l}{\pbobrec({\typed{c}{\chant{\Qbit,\Qbit,\Qbit}},\typed{p}{\chant{\Qbit,\Qbit,\Qbit,\bit,\bit}}})
  = (\qbit s,t)\inp{c}{\typed{x}{\Qbit},\typed{y}{\Qbit},\typed{z}{\Qbit}}\sep}\\
& & \action
{x,s \trans \qgate{CNot}}\sep\action{y,s \trans
  \qgate{CNot}}\sep\action{x,t\trans\qgate{CNot}}\sep\action{z,t
  \trans \qgate{CNot}}\sep\outp{p}{x,y,z,\measure s,\measure
  t}\sep\nil
\end{array}
\]
\[
\begin{array}{rcl}
\multicolumn{3}{l}{\pbobcorr(\typed{p}{\chant{\Qbit,\Qbit,\Qbit,\bit,\bit}},\typed{d}{\chant{\Qbit}})
  = \inp{p}{\typed{x}{\Qbit},\typed{y}{\Qbit},\typed{z}{\Qbit},\typed{j}{\bit},\typed{k}{\bit}}\sep}\\
& &
\action{x\trans\gX^{jk}}\sep\action{y\trans\gX^{j\overline{k}}}\sep\action{z\trans\gX^{\overline{j}k}}\sep\action{x,y\trans\qgate{CNot}}\sep\action{x,z\trans\qgate{CNot}}\sep\outp{d}{x}\sep\nil
\end{array}
\]
\[
\begin{array}{rcl}
\multicolumn{3}{l}{\pname{Bob}(\typed{c}{\chant{\Qbit,\Qbit,\Qbit}},\typed{d}{\chant{\Qbit}})
  = (\new p)(\pbobrec(c,p) \parallel \pbobcorr(p,d))}
\end{array}
\]

The overall effect of the error correcting system is to input a qubit
from channel $a$ and output a qubit, in the same state, on channel
$d$, in the presence of noise. The complete system is defined as
follows.
\[
\pQECC(\typed{a}{\chant{\Qbit}},\typed{d}{\chant{\Qbit}})
  = (\new b,c)(\palice(a,b) \parallel
\pnoise(b,c) \parallel \pbob(c,d))
\]

When we consider correctness of the error correction system, we will
prove that $\pQECC$ is equivalent to the following \emph{identity
  process}, which by definition transmits a single qubit faithfully.
\[
\pidentity(\typed{a}{\chant{\Qbit}},\typed{d}{\chant{\Qbit}}) = \inp{a}{\typed{x}{\Qbit}}\sep\outp{d}{x}\sep\nil
\]

\subsection{Semantics of CQP}
The intended behaviour of the processes in the error correction system
was described informally in the previous section, but in fact the
behaviour is precisely specified by the formal semantics of CQP. In
this section we will explain the formal semantics, although without
giving all of the definitions. Full details can be found in Davidson's
PhD thesis \cite{DavidsonThesis}.


In classical process calculus, the semantics is defined by labelled
transitions between syntactic process terms. For example, a process of
the form $\outp{c}{2}\sep P$, where $P$ is some continuation process,
has the transition
\begin{equation}\label{eq:trans-out}
\ptrns{\outp{c}{2}\sep P}{\outp{c}{2}}{P}.
\end{equation}
The label $\outp{c}{2}$ indicates the potential interaction of the
process with the environment. In order for this potential interaction
to become an actual step in the behaviour of a system, there would
have to be another process, ready to receive on channel $c$. A
suitable process is $\inp{c}{x}\sep Q$, where $Q$ is some continuation
process. The labelled transition representing the potential input is
\begin{equation}\label{eq:trans-in}
\ptrns{\inp{c}{x}\sep Q}{\inp{c}{v}}{Q\subst{v}{x}}.
\end{equation}
Here $v$ stands for any possible input value, and $Q\subst{v}{x}$
means $Q$ with the value $v$ substituted for the variable $x$. If
these two processes are put in parallel then each has a partner for
its potential interaction, and the input and output can synchronize,
resulting in a $\tau$ transition which represents a single step of
behaviour:
\[
\ptrns{\outp{c}{2}\sep P \parallel \inp{c}{x}\sep
  Q}{\tau}{P \parallel Q\subst{2}{x}}.
\]
The complete definition of the semantics takes the form of a
collection of labelled transition rules. Transition
(\ref{eq:trans-out}) becomes a general rule for output if the value $2$ is
replaced by a general value $v$. Transition (\ref{eq:trans-in}) is a
general rule for input. The interaction between input and output is
defined by the rule
\[
\begin{prooftree}
\ptrns{P}{\outp{c}{v}}{P'} \qquad \ptrns{Q}{\inp{c}{v}}{Q'}
\justifies
\ptrns{P \parallel Q}{\tau}{P' \parallel Q'}
\end{prooftree}
\]
which specifies that if the transitions above the line (hypotheses)
are possible then so is the transition below the line (conclusion).
Full details of this style of semantics, in relation to pi-calculus,
can be found in \cite{Milner1999,Sangiorgi2001}.

To define the semantics of a quantum process calculus such as CQP, we
need to include a representation of the quantum state. Because of
entanglement, the quantum state is a global property. It also turns
out to be necessary to specify which qubits in the global quantum
state are owned by (i.e.\ accessible to) the process term under
consideration. We work with \emph{configurations} such as
\begin{equation}\label{eq:quantum-output}
\cnfig{\qstore{q,r}{\frac{1}{\sqrt{2}}(\ket{00}+\ket{11})}}{q}{\outp{c}{q}\sep P}.
\end{equation}
This configuration means that the global quantum state consists of two
qubits, $q$ and $r$, in the specified state; that the process term
under consideration has
access to qubit $q$ but not to qubit $r$ ; and that the process itself
is $\outp{c}{q}\sep P$. Now consider a configuration
with the same quantum state but a different process term:
\[
\cnfig{\qstore{q,r}{\frac{1}{\sqrt{2}}(\ket{00}+\ket{11})}}{r}{\outp{d}{r}\sep Q}.
\]
The parallel composition of these configurations is the following:
\[
\cnfig{\qstore{q,r}{\frac{1}{\sqrt{2}}(\ket{00}+\ket{11})}}{q,r}{\outp{c}{q}\sep P
\parallel \outp{d}{r}\sep Q}
\]
where the quantum state is still the same.

The semantics of CQP consists of labelled transitions between
configurations, which are defined in a similar way to classical
process calculus. For example, configuration (\ref{eq:quantum-output})
has the transition
\[
\ptrns{\cnfig{\qstore{q,r}{\frac{1}{\sqrt{2}}(\ket{00}+\ket{11})}}{q}{\outp{c}{q}\sep
  P}}{\outp{c}{q}}{\cnfig{\qstore{q,r}{\frac{1}{\sqrt{2}}(\ket{00}+\ket{11})}}{\emptyset}{P}}.
\]
The quantum state is not changed by this transition,
but because qubit $q$ is output, the continuation process $P$ no
longer has access to it; the final configuration has an empty list of
owned qubits. 

Previous papers on CQP \cite{Gay2005,Gay2006a} defined the semantics
in a different style. Instead of labelled transitions there were
\emph{reductions}, corresponding to $\tau$ transitions, and these were
defined directly. However, although reduction semantics allows the
behaviour of a complete system to be defined, labelled transitions and
their interpretation as \emph{potential} interactions are necessary in
order to define equivalence between processes, which is the focus of
the present paper.

As well as the different style of definition used in previous work,
there is a very significant difference in the way that the semantics
treats quantum measurement. In the original reduction semantics of
CQP, a measurement leads to a probability distribution over
configurations, which at the next step reduces probabilistically to
one particular configuration. But in order for equivalence of
processes to have the crucial property of \emph{congruence}, the
semantics must incorporate a more sophisticated analysis of
measurement, in which \emph{mixed configurations} play an essential
role.

If the result of a quantum measurement is not made available to an
observer then the system is considered to be in a mixed
state, but it is not sufficient to simply write a mixed quantum
state in a configuration. In general the mixture includes the
process term, because the measurement result occurs within the
term.

\begin{example}
  \label{ex:cfg_measurement}
  $
  ( \qstore{q}{\alpha_0\ket{0} + \alpha_1\ket{1}}; q; \outp{c}{\measure{q}}.P) \transition{\tau} \Dist{i \in \{0,1\}}{\ms{\alpha_i}} ( \qstore{q}{\ket{i}} ; q; \ltrm{x}{\outp{c}{x}.P}{i} )
  $.
\end{example}

This transition represents the effect of a measurement, within a
process which is going to output the result of the measurement; the
output, however, is not part of the transition, which is why it is a
$\tau$ transition and the process term on the right still contains
$\outp{c}{}$.  The configuration on the left is a \emph{pure
  configuration}, as described before.  On the right we have a
\emph{mixed configuration} in which the $\oplus$ ranges over the
possible outcomes of the measurement and the $\ms{\alpha_i}$ are the weights of
the components in the mixture. The quantum state
$\qstore{q}{\ket{i}}$ corresponds to the measurement outcome. The
expression $\ltrmshort{x}{\outp{c}{x}.P}$ is not a $\lambda$-calculus
function, but represents the fact that the components of the mixed
configuration have the same process structure and differ only in the
values corresponding to measurement outcomes. The final term in the
configuration, $i$, shows how the abstracted variable $x$ should be
instantiated in each component. Thus the $\lambda x$ represents a term
into which expressions may be substituted, which is the reason for the
$\lambda$ notation. So the mixed configuration
is essentially an abbreviation of
\[
\ms{\alpha_0} ( \qstore{q}{\ket{0}} ; q; \outp{c}{0}.P\subst{0}{x}) \oplus \ms{\alpha_1} ( \qstore{q}{\ket{1}}; q; \outp{c}{1}.P\subst{1}{x}).
\]

If a measurement outcome is output then it becomes apparent to an
observer which of the possible states the system is in. This is
represented by probabilistic branching, after which we consider that
system to be in one branch or the other --- it is no longer a mixture
of the two. Example~\ref{ex:cfg_output} shows the effect of the output
from the final configuration of Example~\ref{ex:cfg_measurement}. The
output transition produces the intermediate configuration, which is a
probability distribution over pure configurations (in contrast to a
mixed configuration; note the change from $\oplus$ to
$\boxplus$). Because it comes from a mixed configuration, the output
transition contains a \emph{set} of possible values.
From the intermediate configuration there are two possible
probabilistic transitions, of which one is shown ($\ptrans{\ms{\alpha_0}}$). 

\begin{example}\label{ex:cfg_output}
\[
\begin{array}{l}
    \Dist{i \in \{0,1\}}{\ms{\alpha}_i} ( \qstore{q}{\ket{i}} ; q;
    \ltrm{x}{\outp{c}{x}.P}{i} ) \transition{\outp{c}{\{0,1\}}} \\
 \hspace{20ex}\Prob{i \in \{0,1\}}{\ms{\alpha_i}} ( \qstore{q}{\ket{i}} ; q; \ltrm{x}{P}{i} ) 
\ptrans{\ms{\alpha_0}}  ( \qstore{q}{\ket{0}} ; q; \ltrm{x}{P}{0} ) 
\end{array}
\]  
\end{example}

Measurement outcomes may be communicated between processes without
creating a probability distribution. In these cases an observer must
still consider the system to be in a mixed configuration. In
Example~\ref{ex:cfg_communication} there is a mixed configuration on
the left, with arbitrary weights $g_i$, which we imagine to have been
produced by a measurement.  However, there is now a receiver for the
output. Although there is no difference in process $Q$ between the two
components of the mixed configuration, we include it in the $\lambda$
because the communication will propagate the different possible values
for $x$ to $Q$.
\begin{example}
  \label{ex:cfg_communication}
\[
    \Dist{i \in \{0,1\}}{g_i} ( \qstore{q}{\ket{i}} ; q; \ltrm{x}{(\outp{c}{x}.P \parcomp \inp{c}{y}.Q)}{i} ) 
    \transition{\tau} \Dist{i \in \{0,1\}}{g_i} (
    \qstore{q}{\ket{i}} ; q; \ltrm{x}{(P \parcomp Q\subst{x}{y})}{i} )
\]
\end{example}

The full definition of the labelled transition semantics covers more
complex possibilities. For example, if incomplete information about a
measurement is revealed, the resulting configuration is in general a
probability distribution over mixed configurations. The aspects of the
semantics that are relevant to the present paper will be illustrated
further in relation to the error correction example. Now we
define some notation.

 There are two types of transition: probabilistic transitions which
 take the form $\Prob{i}{p_i} s_i \ptrans{p_i} s_i$ where $\forall
 i.(p_i < 1)$, and non-deterministic transitions which have the general
 form $s \transition{\alpha} \Prob{i}{p_i} s_i$ where $\forall i.(p_i
 \leq 1)$ and $\alpha$ is an \emph{action}. The notation $\Prob{i}{p_i}
 s_i \equiv \cfgprob{p_1} s_1 \cfgprobsum \cdots \cfgplus{p_n} s_n$
 denotes a probability distribution over configurations in which
 $\sum_i p_i = 1$. If there is only a single configuration (with
 probability 1) we omit the probability, for example $s
 \transition{\alpha} s'$.

 The separation of probabilistic and non-deterministic transitions
 avoids the need to consider non-deterministic and probabilistic
 transitions from the same configuration. The relations
 $\transition{\alpha}$ and $\ptrans{\pi}$ induce a partition of the
 set 
 \states\ of configurations into non-deterministic configurations \nstates\ and
 probabilistic configurations \pstates: let $\pstates = \{s \in \states
 ~|~ \exists \pi \in (0,1], \exists t \in \states, s \ptrans{\pi} t\}$;
 and let $\nstates = \states \setminus \pstates$. By this definition a
 configuration with no transitions belongs to \nstates. This notation
 will be used in Section~\ref{sec:equivalence}.




\subsection{Execution of  QECC}
We show the interesting steps in one possible execution of QECC,
omitting the $\new$ declarations from the process terms to reduce
clutter. The semantics of CQP is non-deterministic, so transitions can
proceed in a different order; the order shown here is chosen for
presentational convenience. The initial configuration is
$\cnfig{\emptyset}{\emptyset}{\palice\parallel\pnoise\parallel\pbob}$.
In the first few steps, the processes execute $\mathsf{qbit}$ terms,
constructing a global quantum state:
\[
\cnfig{\qstore{y,z,u,v,s,t}{\ket{000000}}}{y,z,u,v,s,t}{\palice'\parallel\pnoise'\parallel\pbob'}
\]
$\palice$ receives qubit $x$, in state $\alpha\ket{0}+\beta\ket{1}$,
from the environment, via transition $\ptrns{}{\inp{a}{x}}{}$, which
expands the quantum state. We now abbreviate the list of qubits to
$\vec{q}=x,y,z,u,v,s,t$. After some $\tau$ transitions corresponding to $\palice$'s
$\qgate{CNot}$ operations, we have:
\[
\cnfig{\qstore{\vec{q}}{\alpha\ket{0000000}+\beta\ket{1110000}}}{\vec{q}}{\outp{b}{x,y,z}\sep\nil\parallel\pnoise'\parallel\pbob'}
\]
$\pnoise' = \pnoiseerr \parallel \pnoisernd'$ ($\pnoisernd'$ has
already done its $\mathsf{qbit}$). The output on $b$ interacts with
the input on $b$ in $\pnoiseerr$. Meanwhile, the measurements in
$\pnoisernd$ produce a mixed configuration because the results are
communicated internally, to $\pnoiseerr$:
\[
\begin{array}{lr}
\multicolumn{2}{l}{\oplus_{j,k\in\{0,1\}}\frac{1}{4}(\qstore{\vec{q}}{\alpha\ket{000jk00}+\beta\ket{111jk00}};\vec{q};}\\
\multicolumn{2}{r}{\hspace{55mm}\ltrm{jk}{\action{x\trans\gX^{jk}}\sep\action{y\trans\gX^{j\overline{k}}}\sep\action{z\trans\gX^{\overline{j}k}}\sep\outp{c}{x,y,z}\sep\nil\parallel\pbob'}{j,k})}
\end{array}
\]
After $\tau$ transitions from the controlled $\gX$ operations, we can
write the mixed configuration explicitly:
\[
\begin{array}{l}
\phantom{\oplus}\frac{1}{4}\cnfig{\qstore{\vec{q}}{\alpha\ket{0000000}+\beta\ket{1110000}}}{\vec{q}}{\outp{c}{x,y,z}\sep\nil\parallel\pbob'}\\
\oplus
\frac{1}{4}\cnfig{\qstore{\vec{q}}{\alpha\ket{0010100}+\beta\ket{1100100}}}{\vec{q}}{\outp{c}{x,y,z}\sep\nil\parallel\pbob'}\\
\oplus
\frac{1}{4}\cnfig{\qstore{\vec{q}}{\alpha\ket{0101000}+\beta\ket{1011000}}}{\vec{q}}{\outp{c}{x,y,z}\sep\nil\parallel\pbob'}\\
\oplus
\frac{1}{4}\cnfig{\qstore{\vec{q}}{\alpha\ket{1001100}+\beta\ket{0111100}}}{\vec{q}}{\outp{c}{x,y,z}\sep\nil\parallel\pbob'}
\end{array}
\]
The remaining transitions operate within the mixed configuration. In
each component of the mixture, the measurement of $s,t$ by $\pbobrec$
has a deterministic outcome, so no further mixedness is
introduced. Eventually we have a mixed configuration in which the
process term is the same, $\outp{d}{x}\sep\nil$, in every component,
so we can just consider the mixed \emph{state}, which is
\[
\oplus_{j,k\in\{0,1\}}\frac{1}{4}\qstore{x,y,z,u,v,s,t}{\alpha\ket{000jkjk}+\beta\ket{100jkjk}}.
\]
The mixture over $j,k$ is the residue of the random choice made by
$\pnoisernd$, and the dependence of $s$ and $t$ on $j,k$ is because
$\pbobrec$'s measurement recovers the values of $j$ and $k$ (which is
what allows the error to be corrected). In this final mixed state, the
reduced density matrix of $x$, which is what we are interested in when
$x$ is output, is the same as the original density matrix of $x$.

\section{Behavioural Equivalence of CQP Processes}
\label{sec-equivalence}
\label{sec:equivalence}

The process calculus approach to verification is to define a process
$\pname{System}$ which models the system of interest, another process
$\pname{Spec}$ which expresses the specification that $\pname{System}$
should satisfy, and then prove that $\pname{System}$ and
$\pname{Spec}$ are equivalent. Usually $\pname{Spec}$ is defined in a
sufficiently simple way that it can be taken as self-evident that it
accurately represents the desired specification.

What do we mean by \emph{equivalent}? The idea is that two processes
are equivalent if their behaviour is indistinguishable by an observer.
That is, if they do the same thing in the same circumstances.
Equivalence relations in this style are generically called
\emph{behavioural} equivalences. Suppose that $\cong$ is an
equivalence relation on processes. The ideal situation is for $\cong$
to have a further property called \emph{congruence}, which means that
it is preserved by all of the constructs of the process calculus. A
convenient way to express this property involves the notion of a
\emph{process context} $\ctxt{C}{}$. This is a process term containing a
\emph{hole}, represented by $[]$, into which a process term may be
placed. For example, $\inp{c}{x}\sep[]$ is a context, and putting the
process $\outp{d}{x}\sep{\nil}$ into the hole results in the process
$\inp{c}{x}\sep\outp{d}{x}\sep{\nil}$.

\begin{definition}
An equivalence relation $\cong$ on processes is a \emph{congruence} if
\[
\forall P, Q.~ P\cong Q \Rightarrow \forall \ctxt{C}{}.~
\ctxt{C}{P}\cong\ctxt{C}{Q}.
\] 
\end{definition}

This definition of congruence corresponds to the idea that observers
are themselves expressed as processes. Congruence, in addition to the
property of being an equivalence relation, is what is required in
order to allow equational reasoning about equivalence of processes. It
means that if a system satisfies its specification, then it continues
to satisfy its specification no matter what environment it is placed
in.

From the beginning of the study of quantum process calculus, the aim
was to define a behavioural equivalence with the congruence property.
This was not straightforward and took several years to achieve; Lalire
\cite{Lalire2006} describes an unsuccessful attempt. Recently the
congruence problem has been solved by the first three authors of the
present paper \cite{DavidsonThesis} for CQP and, independently, by
Feng \emph{et al.} \cite{Feng2011} for qCCS.

We will now present the concept of \emph{bisimilarity}, which is the
main approach to behavioural equivalence, and then define a particular
form of bisimilarity, called \emph{probabilistic branching
  bisimilarity}, which is a congruence for CQP.  


\subsection{Strong Bisimilarity}
The basic idea of bisimilarity is that if two processes are equivalent
then any labelled transition by one can be matched by the other, and
the resulting processes are again equivalent. It is worth presenting
the definition of the prototypical example, \emph{strong bisimilarity}
\cite{Milner1989}, as a model for later definitions. The most general setting
for the definition is to consider a \emph{labelled transition system},
which consists of a set of states and a three-place relation on
$\mathit{States}\times\mathit{Labels}\times\mathit{States}$, written
$\ptrns{s}{\alpha}{t}$. A labelled transition system can be regarded as a
labelled directed graph whose vertices are the states. We will
consider relations on the set of states. The definition of strong
bisimilarity proceeds in two stages. First we define the property of
\emph{strong bisimulation}, which a particular relation might or might
not have.
\begin{definition}[Strong Bisimulation]
  A relation $\mathcal{R}$ is a \emph{strong bisimulation} if whenever
  $(P,Q) \in \mathcal{R}$ then for all labels $\alpha$, both
  \begin{enumerate}
    \item if $P \transition{\alpha} P'$ then $Q \transition{\alpha}
      Q'$ and $(P',Q') \in \mathcal{R}$, and
    \item if $Q \transition{\alpha} Q'$ then $P \transition{\alpha} P'$ and $(P',Q') \in \mathcal{R}$.
   \end{enumerate}
\end{definition}
For a given labelled transition system there are many relations that
have the property of strong bisimulation, including (trivially) the
empty relation. The key idea is to define
\emph{strong bisimilarity} to be the union of all strong
bisimulations, or equivalently the largest strong bisimulation.
In other words, $P$ and $Q$ are strong bisimilar (denoted $P \sim Q$) if
and only if there exists a bisimulation $\mathcal{R}$ such that $(P,Q)
\in \mathcal{R}$.

\subsection{Probabilistic Branching Bisimilarity}
One of the characteristics of strong bisimilarity is that it is a
stronger relation than trace equivalence; it is possible for two
processes to generate the same sequences of labels, but not be strong
bisimilar. Strong bisimilarity depends on the branching structure of
the processes as well as on their sequences of labels. Another
characteristic is that \emph{every} transition must be matched
exactly, including $\tau$ transitions. However, because they arise
from internal communications, it is often undesirable to insist that
equivalent processes must match each other's $\tau$ transitions. Hence
weaker variations of bisimilarity have been defined, including
\emph{weak bisimilarity} \cite{Milner1989}, which ignores $\tau$
transitions, and \emph{branching bisimilarity} \cite{Glabbeek1996},
which reduces the significance of $\tau$ transitions but retains
information about their branching structure.

When considering equivalences for quantum process calculus, it is
necessary to take probability into account; even with the treatment of
mixed configurations described in Section~\ref{sec:CQP}, there is
probabilistic behaviour when measurement results are revealed to the
observer. There are several varieties of probabilistic bisimilarity
for classical probabilistic process calculi, including
\emph{probabilistic branching bisimilarity} \cite{Trcka2008}. The equivalence
for CQP defined by Davidson \cite{DavidsonThesis}, which turns out to be a
congruence, is a form of probabilistic branching bisimilarity, adapted
to the situation in which probabilistic behaviour comes from quantum
measurement. A key point is that when considering matching of input or
output transitions involving qubits, it is the reduced density
matrices of the transmitted qubits that are required to be equal.

Although we did not present the full definition of the labelled
transition semantics for CQP, we will now define probabilistic
branching bisimilarity in full. In Section~\ref{sec:errorcorrectionequivalence}, the definition will
be applied to the error correction example. The 
definitions in the remainder of this section are from Davidson's
thesis \cite{DavidsonThesis}.


\textbf{Notation:} Let $\opttrans{\tau}$ denote zero
or one $\tau$ transitions; let $\weaktrans{ }$ denote zero or more $\tau$
transitions; and let $\weaktrans{\alpha}$ be equivalent to $\weaktrans{ }
\transition{\alpha} \weaktrans{ }$. We write $\vec{q}$ for a list of
qubit names, and similarly for other lists.


\begin{definition}[Density Matrix of Configurations]
  \label{def:density_matrix_configs}
  Let $\sigma_i = \qstore{\vec{p}}{\ket{\psi_i}}$ and $\vec{q}
  \subseteq \vec{p}$ and $s_i = (\sigma_i; \omega;
  \ltrm{\vec{x}}{P}{\vec{v}_i})$ and $s = \Dist{i}{g_i} s_i$. Then
\[
\begin{array}{llcll}
1. & \rho(\sigma_i) = \ketbra{\psi_i}{\psi_i} & ~~~~~~ & 4. &
\rho^{\vec{q}}(s_i) = \rho^{\vec{q}}(\sigma_i) \\
2. & \rho^{\vec{q}}(\sigma_i) =
\trace_{\vec{p}\setminus\vec{q}}(\ketbra{\psi_i}{\psi_i}) & & 5. &
\rho(s) = \sum_i g_i \rho(s_i) \\
3. & \rho(s_i) = \rho(\sigma_i) & & 6. & \rho^{\vec{q}}(s) = \sum_i g_i \rho^{\vec{q}}(s_i)
\end{array}
\]
\end{definition}

We also introduce the notation $\rhoe$ to denote the reduced density
matrix of the \emph{environment} qubits. Formally, if $s =
(\qstore{\vec{q}}{\ket{\psi}}; \vec{p}; P)$ then $\rhoe(s) =
\rho^{\vec{r}}(s)$ where $\vec{r} = \vec{q} \setminus \vec{p}$. The
definition of $\rhoe$ is extended to mixed configurations in the same
manner as $\rho$.

Again let $\states$ be the set of configurations. 
The probabilistic function $\mu: \states \times \states \rightarrow
[0,1]$ is defined in the style of \cite{Trcka2008}. It allows
non-deterministic transitions to be treated as transitions with
probability $1$, which is necessary when calculating the total
probability of reaching a terminal state.
$\mu(s,t) = \pi$ if $s \ptrans{\pi} t$; $\mu(s,t) = 1$ if $s = t$ and
$s \in \nstates$; $\mu(s,t) = 0$ otherwise.


\begin{definition}[Probabilistic Branching Bisimulation]
  \label{def:pbb}
  An equivalence relation $\mathcal{R}$
  on configurations is a \emph{probabilistic branching bisimulation} on configurations if whenever $(s,t) \in \mathcal{R}$ the following conditions are satisfied.
  \begin{enumerate}[I.]
	    \item If $s \in \nstates$ and $s \transition{\tau} s'$ 
	      then $\exists t', t''$ such that $t \weaktrans{ } t'
              \opttrans{\tau} t''$ with $(s, t') \in \mathcal{R}$ and $(s', t'') \in \mathcal{R}$.
	    \item If $s \transition{\outp{c}{V,\vec{q}_1}} s'$ where $s' = \Prob{j \in \{1\dots m\}}{p_j} s_j'$ and $V = \{\vec{v}_1,\dots,\vec{v}_m\}$ then $\exists t', t''$ such that $t \weaktrans{ } t' \transition{\outp{c}{V,\vec{q}_2}} t''$ with
	      \begin{enumerate}[a)]
		\item $(s, t') \in \mathcal{R}$,
		\item $t'' = \Prob{j \in \{1\dots m\}}{p_j} t_j''$,
                \item for each $j \in \{1,\dots,m\}$, $\rhoe(s_j') = \rhoe(t_j'')$.
		\item for each $j \in \{1,\dots,m\}$, $(s_j', t_j'') \in \mathcal{R}$.
	      \end{enumerate}
        \item If $s \transition{\inp{c}{\vec{v}}} s'$ then $\exists
          t', t''$ such that $t \weaktrans{ } t'
          \transition{\inp{c}{\vec{v}}} t''$ with $(s, t') \in
          \mathcal{R}$ and $(s', t'') \in \mathcal{R}$.
	    \item If $s \in \pstates$ then $\mu(s, D) = \mu(t, D)$ for all classes $D \in \states/\mathcal{R}$.
    \end{enumerate}
\end{definition}

This relation follows the standard definition of branching
bisimulation \cite{Glabbeek1996} with additional conditions for
probabilistic configurations and matching quantum information. In
condition II we require that the distinct set of values $V$ must match
and although the qubit names ($\vec{q}_1$ and $\vec{q}_2$) need not be
identical, their respective reduced density matrices
($\rho^{\vec{q}_1}(s)$ and $\rho^{\vec{q}_2}(t')$) must.

Condition IV provides the matching on probabilistic configurations
following the approach of \cite{Trcka2008}. In this relation, a
probabilistic configuration which necessarily evolves from an output
will satisfy IV if the prior configuration satisfies II d). It is
necessary to include the latter condition to ensure that the
probabilities are paired with their respective configurations.

Naturally this leads to the following definition of bisimilarity on
configurations.
\begin{definition}[Probabilistic Branching Bisimilarity]
Configurations $s$ and $t$ are \emph{probabilistic branching bisimilar}, denoted $s \pbsim t$, if there exists a probabilistic branching bisimulation $\mathcal{R}$ such that $(s,t) \in \mathcal{R}$.
\end{definition}

What we really want is equivalence of processes, independently of
configurations (i.e.\ independently of particular quantum states).

\begin{definition}[Probabilistic Branching Bisimilarity of Processes]
Processes $P$ and $Q$ are \emph{probabilistic branching bisimilar}, denoted $P \pbsim Q$, if and only if for all $\sigma$, $(\sigma; \emptyset; P) \pbsim (\sigma; \emptyset; Q)$.
\end{definition}

For convenience, in the remainder of this paper \emph{bisimilarity}
will refer to probabilistic branching bisimilarity and it will be
clear from the context whether this is the relation on processes or
configurations. The same symbol, $\pbsim$, is used for both relations.

It turns out that probabilistic branching bisimilarity is not a
congruence because it is not preserved by substitution of values for
variables, which is significant because of the use of substitution to
define the semantics of input. We therefore define a stronger
relation, \emph{full probabilistic branching bisimilarity}, which is
the closure of probabilistic branching bisimilarity under substitutions.

\begin{definition}[Full probabilistic branching bisimilarity]
  Processes $P$ and $Q$ are \emph{full probabilistic branching
    bisimilar}, denoted $P \fpbsim Q$, if for all substitutions $\kappa$
and all quantum states
  $\sigma$, $(\sigma; \vec{q}; P\kappa) \pbsim (\sigma; \vec{q};
  Q\kappa)$. 
\end{definition}

We are now able to state the main result of \cite{DavidsonThesis}.
\begin{theorem}[Full probabilistic branching bisimilarity is a congruence]
  \label{thm:congruence}
  If $P \fpbsim Q$ then for any context $\ctxt{C}{}$, if
  $\ctxt{C}{P}$ and $\ctxt{C}{Q}$ are typable then $\ctxt{C}{P}
  \fpbsim \ctxt{C}{Q}$. 
\end{theorem}
The condition that $\ctxt{C}{P}$ and $\ctxt{C}{Q}$ are typable is
used to ensure that the context does not manipulate qubits that are
owned by $P$ or $Q$.

\subsection{Mixed Configurations and Congruence}
A simple example will illustrate why the congruence result depends
crucially on the use of mixed configurations. Consider the processes
\[
\begin{array}{rcl}
P(\typed{a}{\chant{\Qbit}}) & = &
\inp{a}{\typed{x}{\Qbit}}\sep\action{\measure x}\sep\nil \\
Q(\typed{a}{\chant{\Qbit}}) & = & \inp{a}{\typed{x}{\Qbit}}\sep\action{x\trans\gH}\action{\measure x}\sep\nil 
\end{array}
\]
$P$ and $Q$ are probabilistic branching bisimilar, because in any
quantum state they can match each other's transitions. For input
transitions this is because they can both input a single qubit, and
for output transitions it is trivial because neither process produces
any output. The actions within each process produce $\tau$
transitions, which are absorbed into the input transitions according
to the definition of probabilistic branching bisimulation.

Now consider $P$ and $Q$ in parallel with
$R(\typed{b}{\chant{\Qbit}}) = \outp{b}{q}\sep\nil$
in the quantum state
$\qstore{p,q}{\frac{1}{\sqrt{2}}(\ket{00}+\ket{11})}$. That is,
consider the configurations
\[
\begin{array}{ccc}
\cnfig{\qstore{p,q}{\frac{1}{\sqrt{2}}(\ket{00}+\ket{11})}}{p,q}{P\parallel
R} & ~~~~~~ & 
\cnfig{\qstore{p,q}{\frac{1}{\sqrt{2}}(\ket{00}+\ket{11})}}{p,q}{Q\parallel
R}
\end{array}
\]
The interesting situation is when the measurement in $P$ or $Q$ occurs
before the output in $R$.  Imagine, first, that the semantics of CQP
handles the measurement by producing a probability distribution;
recall that this is not the actual semantics of measurement.  In
$P\parallel R$ the quantum state before the measurement is
$\frac{1}{\sqrt{2}}(\ket{00}+\ket{11})$ and the state after the
measurement is either $\ket{00}$ or $\ket{11}$ with equal
probability. The qubit output by $R$ has reduced density matrix
$\matr{1}{0}{0}{0}$ or $\matr{0}{0}{0}{1}$. In $Q\parallel R$ the
quantum state before the measurement is
$\frac{1}{2}(\ket{00}+\ket{01}+\ket{10}-\ket{11})$ and the state after
the measurement is either $\frac{1}{\sqrt{2}}(\ket{00}+\ket{01})$ or
$\frac{1}{\sqrt{2}}(\ket{10}-\ket{11})$ with equal probability. The
qubit output by $R$ has reduced density matrix
$\frac{1}{2}\matr{1}{1}{1}{1}$ or $\frac{1}{2}\matr{1}{-1}{-1}{1}$. It
is therefore impossible for $P\parallel R$ and $Q\parallel R$ to match
each other's outputs.

Actually, of course, the semantics of CQP does not produce a
probability distribution in this case, because the result of the
measurement is not output. Instead, from $P\parallel R$ we get the
mixed configuration
\begin{equation}
\frac{1}{2}\cnfig{\qstore{p,q}{\ket{00}}}{p,q}{\nil\parallel\outp{b}{q}\sep\nil}
\oplus \frac{1}{2}\cnfig{\qstore{p,q}{\ket{11}}}{p,q}{\nil\parallel\outp{b}{q}\sep\nil}
\end{equation}
and from $Q\parallel R$ we get the mixed configuration
\begin{equation}
\frac{1}{2}\cnfig{\qstore{p,q}{\frac{1}{\sqrt{2}}(\ket{00}+\ket{01})}}{p,q}{\nil\parallel\outp{b}{q}\sep\nil}
\oplus \frac{1}{2}\cnfig{\qstore{p,q}{\frac{1}{\sqrt{2}}(\ket{10}-\ket{11})}}{p,q}{\nil\parallel\outp{b}{q}\sep\nil}.
\end{equation}
The calculation of the reduced density matrix of the qubit output by
$R$, taking into account the contributions of each component of the
mixed configuration, gives $\matr{1}{0}{0}{1}$ in both cases. This
enables $P\parallel R$ and $Q\parallel R$ to match each other's
outputs, and in fact (although we do not show it here), it is
straightforward to construct a probabilistic branching bisimulation
relation containing $(P\parallel R, Q\parallel R)$.

\subsection{Correctness of QECC}
\label{sec:errorcorrectionequivalence}
We now sketch the proof that $\pQECC \fpbsim \pidentity$, which by
Theorem~\ref{thm:congruence} implies that the error correction system
works in any context. An interesting consequence is that the qubit
being transmitted may be part of any quantum state, meaning that it is
correctly transmitted with error correction even if it is entangled
with other qubits; the entanglement is also preserved by the error
correction system. This property of error correction, although easily
verified by hand, is not usually stated explicitly in the literature.

\begin{proposition}
  $\pQECC \fpbsim \pidentity$.
\end{proposition}
 \begin{proof}
 First we prove that $\pQECC \pbsim \pidentity$,
   by defining an equivalence relation $\mathcal{R}$ that contains the pair
   $(\cnfig{\sigma}{\emptyset}{\pQECC},\cnfig{\sigma}{\emptyset}{\pidentity})$
   for all $\sigma$ and is closed under their transitions. 
   $\mathcal{R}$ is defined by taking its equivalence classes to be the
   $S_i(\sigma)$ defined below, for all states $\sigma$. The idea is to
   group configurations according to the sequences of observable
   transitions leading to them. $S_2$ is also
   parameterized by the input qubit, as this affects the output qubit
   and hence the equivalence class.
 \[
 \begin{array}{rcl}
 \pname{S_{1}(\sigma)} & = & \{s \mid
 (\sigma;\emptyset;P)\weaktrans{}s ~\mbox{and}~ P \in \{\pQECC,\pidentity\}\}\\
 \pname{S_{2}(\sigma,x)} & = & \{s  \mid  (\sigma;\emptyset;P)
 \weaktrans{\inp{a}{x}}s ~\mbox{and}~ P \in \{\pQECC,\pidentity\}\} \\
 \pname{S_{3}(\sigma)} & = & \{s  \mid  (\sigma;\emptyset;P)
 \weaktrans{\inp{a}{x}}\weaktrans{\outp{d}{x}}s ~\mbox{and}~ P \in \{\pQECC,\pidentity\}\}
 \end{array}
 \]
 We now prove that $\mathcal{R}$ is a probabilistic branching
 bisimulation. It suffices to consider transitions between $S_i$
 classes, as transitions within classes must be $\tau$ and are matched
 by $\tau$.  If $s, t \in S_{1}(\sigma)$ and
 $s\transition{\inp{a}{x}}s'$ then $s'\in S_{2}(\sigma)$ and we find
 $t',t''$ such that $t\weaktrans{}t'\transition{\inp{a}{x}}t''$ with
 $t'\in S_{1}(\sigma)$ and $t''\in S_{2}(\sigma)$, so $(s,t')\in
 \mathcal{R}$ and $(s',t'')\in\mathcal{R}$ as required. Transitions
 from $S_2(\sigma)$ are matched similarly. There are no
 transitions from $S_3(\sigma)$.


There is no need for a probability calculation (case IV of
Definition~\ref{def:pbb}) because no
probabilistic configurations arise; measurement results are always
communicated internally, and never to the external environment.

Finally, because $\pQECC$ and $\pidentity$ have no free variables,
their equivalence is trivially preserved by substitutions.\qed
\end{proof}
 


\section{Error Correction: A Second Model}
\label{sec-example}
\label{sec:example}
We now consider a different noise model in which random $\gX$ errors
are applied independently to each of the three qubits being
transmitted. The new definition of $\pnoise$ is shown below; we use
the original definitions of $\palice$ and $\pbob$; the overall system
is now $\pQECC2$.
\[
\begin{array}{ll}
\multicolumn{2}{l}{\pnoisernd(\typed{p}{\chant{\bit,\bit,\bit}})
=}\\
~ & (\qbit
u,v,w)\sep\action{u\trans\gH}\sep\action{v\trans\gH}\sep\action{w\trans\gH}\sep\outp{p}{\measure
  u,\measure v,\measure w}\sep\nil \\
\multicolumn{2}{l}{\pnoiseerr(\typed{b}{\chant{\Qbit,\Qbit,\Qbit}},\typed{p}{\chant{\bit,\bit,\bit}},\typed{c}{\chant{\Qbit,\Qbit,\Qbit}})
  =}\\
& \inp{b}{\typed{x}{\Qbit},\typed{y}{\Qbit},\typed{z}{\Qbit}}\sep\inp{p}{\typed{j}{\bit},\typed{k}{\bit},\typed{l}{\bit}}\sep\action{x\trans\gX^{j}}\sep\action{y\trans\gX^{k}}\sep\action{z\trans\gX^{l}}\sep\outp{c}{x,y,z}\sep\nil \\
\multicolumn{2}{l}{\pnoise(\typed{b}{\chant{\Qbit,\Qbit,\Qbit}},\typed{c}{\chant{\Qbit,\Qbit,\Qbit}})
  = (\new p)(\pnoisernd(p) \parallel \pnoiseerr(b,p,c))}
\end{array}
\]
\[
\begin{array}{rcl}
\pQECC2(\typed{a}{\chant{\Qbit}},\typed{d}{\chant{\Qbit}}) & = &
(\new b,c)(\palice(a,b) \parallel
\pnoise(b,c) \parallel \pbob(c,d))
\end{array}
\]
The threefold repetition code is not able to correct multiple errors,
so we do not have $\pQECC2 \fpbsim \pidentity$. The error correction
system has a probability of $\frac{1}{2}$ of transmitting a qubit with
an $\gX$ error. We can express this in CQP by using $\pbitflip$ as a
specification process:
\[
\begin{array}{rcl}
\prnd(\typed{p}{\chant{\bit}}) & = & (\qbit u)\action{u\trans\gH}\sep\outp{p}{\measure u}\sep\nil \\
\pflip(\typed{a}{\chant{\Qbit}},\typed{p}{\chant{\bit}},{\typed{d}{\chant{\Qbit}}})& = & \inp{a}{\typed{x}{\Qbit}}\sep\inp{p}{\typed{j}{\bit}}\sep\action{x\trans\gX^i}\sep\outp{d}{x}\sep\nil\\
\pbitflip(\typed{a}{\chant{\Qbit}},\typed{d}{\chant{\Qbit}})  & = &
(\new p)(\prnd(p) \parallel\pflip(a,p,d))
\end{array}
\]
and by very similar arguments to before, we obtain:
\begin{proposition}
$\pQECC2 \fpbsim \pbitflip$.
\end{proposition}


There is still no probability calculation because the results of the
measurements in $\pnoisernd$ and $\prnd$ are not output. The equal
probability of correct and incorrect transmission manifests itself in
the fact that the reduced density matrix of the final output qubit,
from both $\pQECC2$ and $\pbitflip$, is an equal mixture of the input
qubit and its inverse. The only way to introduce probability into
this example is for $\pflip$ to observably output $j$ and $\pnoiseerr$
to observably output the majority value of $j,k,l$, before the final
qubit output. 

We know from the standard analysis of this error correction system
that if the independent probability of flipping each qubit is
$p<\frac{1}{2}$, $\pQECC2$ reduces the overall probability of a
bit-flip error to $p^2(3-2p) < p$. With a slightly more complicated
analysis we could also express this property in CQP.

\section{Conclusion and Future Work}
\label{sec-conclusion}
\label{sec:conclusion}
We have explained the use of the process calculus CQP, and its theory
of behavioural equivalence, in analyzing the correctness of quantum
communication systems. We have summarized the theory, which
is presented in full detail in \cite{DavidsonThesis}, and given two
examples based on a simple quantum error correcting code.

Quantum error correction can easily be analyzed by pen and paper, but
the point of process calculus is that it forms part of a systematic
methodology for verification of quantum systems. In particular, the
congruence property of behavioural equivalence explicitly guarantees
that equivalent processes remain equivalent in any context, and
supports equational reasoning. For example: we have shown that $\pQECC
\fpbsim \pidentity$; there is a proof in \cite{DavidsonThesis} that
$\pteleport \fpbsim \pidentity$; so we have, for free, that
$\pQECC \fpbsim \pteleport$, in any context. Because CQP can also
express classical behaviour, we have a uniform framework in which to
analyze classical and quantum computation and communication.

The next steps for this line of research are to develop equational
axiomatizations of behavioural equivalence, in order to reduce the
need to explicitly construct bisimulation relations, and to develop
software for automatic verification of equivalence. Both of these
techniques are already well established for classical process
calculus.

\bibliographystyle{eptcs}

\bibliography{main}

\end{document}